\documentclass[11pt, onecolumn]{article}
\usepackage[top=1in, bottom=1in, left=1.25in, right=1.25in]{geometry}

\usepackage{amsfonts}
\usepackage{amsmath}
\usepackage{amssymb}
\usepackage{graphicx}
\usepackage{color, soul}

\usepackage[amsmath,thmmarks]{ntheorem}
\usepackage{theorem}

\newtheorem{definition}{Definition}
\newtheorem{theorem}{Theorem}
\newtheorem{lemma}{Lemma}

\newtheorem{remark}{Remark}

\theoremheaderfont{\sc}\theorembodyfont{\upshape}
\theoremstyle{nonumberplain}
\theoremseparator{}
\theoremsymbol{\rule{1ex}{1ex}}
\newtheorem{proof}{Proof}

\newcommand{\T}{\mathrm{T}}
\newcommand{\GL}{\mathrm{Gr}}

\newcommand{\figurewidth}{0.7\textwidth}

\begin{document}
\title{Restricted Isometry Property of Subspace Projection Matrix Under Random Compression}

\author{Xinyue~Shen~and~Yuantao~Gu
\thanks{The authors are with the Department of Electronic Engineering, Tsinghua University, Beijing 100084, China. The corresponding author of this work is Yuantao Gu (e-mail: gyt@tsinghua.edu.cn).}}%

\date{Submitted Dec 8, 2014, revised Feb 4, 2015; accepted Feb 6, 2015;\\
to appear in \emph{IEEE Signal Processing Letters}}

\maketitle

\begin{abstract}
Structures play a significant role in the field of signal processing.
As a representative of structural data, low rank matrix along with its restricted isometry property (RIP) has been an important research topic in compressive signal processing.
Subspace projection matrix is a kind of low rank matrix with additional structure, which allows for further reduction of its intrinsic dimension. This leaves room for improving its own RIP, which could work as the foundation of compressed subspace projection matrix recovery. In this work, we study the RIP of subspace projection matrix under random orthonormal compression. Considering the fact that subspace projection matrices of $s$ dimensional subspaces in $\mathbb{R}^N$ form an $s(N-s)$ dimensional submanifold in $\mathbb{R}^{N\times N}$, our main concern is transformed to the stable embedding of such submanifold into $\mathbb{R}^{N\times N}$. The result is that by $O(s(N-s)\log N)$ number of random measurements the RIP of subspace projection matrix is guaranteed.

\end{abstract}

\section{Introduction}
\label{sec:intro}
Signal structure has always been a key point in the field of signal processing.
Structural data, such as sparse signal and low rank matrix, have been important research topics in compressive signal processing \cite{candes2008intro,Fazel2002}.
These structures invoke low intrinsic dimension, so the restricted isometry property (RIP) can be established to guarantee both exact and robust reconstructions from randomly compressed measurements \cite{baraniuk2008simple,LRRIPNNM,LRRIP}.

For a given $s$ dimensional linear subspace $\mathcal{S}$ in an Euclidean space $\mathbb{R}^N$, assuming that $s$ is less than $N$, the subspace projection matrix is a low rank matrix with rather specific structure. In fact, it is not only symmetric, semi-definite, but also has merely eigenvalues 1 and 0. Such additional structure invokes lower intrinsic dimension than a general low rank matrix does, therefore theoretical improvement on the RIP can be expected.

According to the basic ideas in compressive sensing \cite{RIPCS}, the RIP of subspace projection matrix could work as the foundation of compressed subspace projection matrix recovery. Considering the fact that subspace projection matrix has a one to one correspondence with subspace, the recovery from its compression could be viewed as compressed subspace estimation. Subspace estimation has been a concerning problem in signal processing and computer vision. In some scenarios, such as face recognition \cite{LSFC}, motion segmentation \cite{MoSeg}, and visual tracking \cite{1315111}, the objects belong to subspaces with much lower dimension than the ambient space. In fact, subspace estimation from highly incomplete information has recently appeared as an attractive research topic \cite{OITSHII,PETRELS,CSCCaseStudy}.

For subspaces with a given dimension in $\mathbb{R}^N$, their projection matrices form a manifold. There are significant and solid works in manifold-modeled signal recovery from randomly compressed measurements \cite{StaManEmb,ManEmbNew,ManRanLProj,RanProjMan}. These works extend classic compressed sensing by generalizing low-dimension model from sparse signal to signal on low dimensional manifold, and study stable manifold embeddings and nonadaptive dimensionality reduction of data on manifold.
One of the key ideas is to control the regularity of the manifold so that it is well-conditioned. The work in \cite{RanProjManIMPORTANT} utilizes an instructive quantity called the condition number of a manifold also known as the reach of a manifold, which studies submanifold extrinsically and unveils its Riemannian geometry properties \cite{HomologyRanSam}.

In this work, we aim to study the RIP of subspace projection matrix under random orthonormal compression.
A matrix manifold is used to model the set of subspace projection matrices.
By investigating the differential structure and the condition number of such manifold, we are able to conclude that by $O(s(N-s)\log N)$ random measurements the RIP of subspace projection matrices is guaranteed.

\section{Main result}
\label{sec:main}
In this work, we study the RIP of subspace projection matrices under random orthonormal compression.
\begin{definition}
The set of projection matrices corresponding to $s$ dimensional subspaces in $\mathbb{R}^{N}$ is defined as
\begin{align}\label{projm}
\mathcal{P}_{N,s}:&=\{P_X= X(X^\T X)^{-1}X^\T:
 X\in\mathbb{R}^{N\times s}, \mathrm{dim}(\mathrm{span}(X))=s\} \nonumber \\
&=\{ P_X=XX^\T:\;X\in\GL_{N,s}\},
\end{align}
in which $\mathrm{span}(X)$ denotes the column space of $X$, and $\GL_{N,s}$ is the Grassmann manifold of $s$ dimensional subspaces in $\mathbb{R}^N$.
\end{definition}

\begin{remark}
Equation \eqref{projm} is obtained by ortho-normalizing the columns of $X$ while keeping $\mathrm{span}(X)$ fixed. Because different choices of $X$ do not change $P_X$ as long as $\mathrm{span}(X)$ is fixed, we have $X\in\GL_{N,s}$ \cite{OptAl}.
\end{remark}
\begin{remark}
From Definition \ref{projm}, we know that, for a linear subspace $\mathcal{S}$, its projection matrix is the matrix that a vector ${\bf x}\in\mathbb{R}^N$ has to multiply when projected onto $\mathcal{S}$. This is the reason that it is called a projection matrix.
\end{remark}

$\mathcal{P}_{N,s}$ is an $s(N-s)$ dimensional submanifold in $\mathbb{R}^{N\times N}$.
The following theorem describes the RIP of $\mathcal{P}_{N,s}$ under random orthoprojector.

\begin{theorem}\label{RIPP}
For fixed $0<\varepsilon<1$ and $\beta>0$, assume that $N\geq 3$. Let $\mathcal{A}:\mathbb{R}^{N\times N}\rightarrow\mathbb{R}^m$ be a random orthoprojector with
\begin{align}
m\geq
&\left( \frac{2+\beta}{\varepsilon^2-\varepsilon^3/3}\right) O\left(s(N-s)\log \frac{N}{\varepsilon}\right).
\end{align}
If $m<N^2$, then with probability exceeding $1-\mathrm{e}^{-c_1\beta s(N-s)}$, in which $c_1$ is a universal constant, the following property holds for every pair of $P_X,P_Y\in\mathcal{P}_{N,s}$, $P_X\neq P_Y$,
\begin{equation}
(1-\varepsilon)\frac{\sqrt{m}}{N} \leq \frac{\|\mathcal{A}(P_X-P_Y)\|_2}{\|P_X-P_Y\|_F} \leq (1+\varepsilon)\frac{\sqrt{m}}{N}.
\end{equation}
\end{theorem}
\begin{proof}
The proof is postponed to section \ref{sec:proof}.
\end{proof}

From Theorem \ref{RIPP}, the number of measurements $m\geq O(s(N-s)\log N)$ is enough to guarantee the RIP under random orthonormal compression. For low rank matrices without further specific structure, the number of measurements should be no less than $O(sN\log N)$ \cite{LRRIPNNM}. One may notice that when $s\ll N$, the improvement from the latter to the former is not much. Although it is true in that case, the result in Theorem \ref{RIPP} does improve the scaling law of the number of measurements on $s$, and verifies the intuition that, compared with low rank matrices, subspace projection matrices have additional structure which is able to further reduce the number of compressed measurements needed for reconstruction.

Heuristically, the scaling law $m\geq O((N-s)s\log N)$ is reasonable, in that the degree of freedom of a subspace projection matrix is $s(N-s)$. It should be highlighted that although Theorem \ref{RIPP} is for an orthonormal random compression, the conclusion could be naturally extended to random compressions satisfying the concentration property \cite{Coin}.

In the following text, the establishment of Theorem \ref{RIPP} will be demonstrated. The first order differential structure and the condition number of $\mathcal{P}_{N,s}$ are studied in section \ref{sec:projm}, and then Theorem \ref{RIPP} is readily proved in section \ref{sec:proof}.

\section{Manifold of Subspace Projection Matrix}
\label{sec:projm}
Matrix manifold has been a powerful tool to structural matrix data recovery \cite{RiePurMatRec}. Typical matrix manifolds, such as the Grassmann manifold, the orthogonal group, and the Stiefel manifold,  have been comprehensively studied, and one may read \cite{JohnLee,OptAl,GrassGeom} for reference.

In this section, we study the set of subspace projection matrices of $s$ dimensional subspaces in $\mathbb{R}^{N}$ denoted as $\mathcal{P}_{N,s}$ and defined in \eqref{projm}.
Because a subspace has a one to one correspondence with a projection matrix, and such correspondence is continuous, $\mathcal{P}_{N,s}$ is an $s(N-s)$ dimensional manifold which is homeomorphic to $\GL_{N,s}$.

Preceding the calculation of the differential structure and the condition number, we may first illustrate $\mathcal{P}_{N,s}$ for specific $N=2$ and $s=1$. From the definition, we know that
\begin{align*}
\mathcal{P}_{2,1}
&\cong \left\{[x_1^2,\sqrt{2}x_1x_2,x_2^2]^\T:\;[x_1,x_2]^\T\in\GL_{2,1}\right\}.
\end{align*}
Because $\GL_{2,1}\cong\mathbb{S}^1$, $\mathcal{P}_{2,1}$ can be expressed as a circle with radius $1/\sqrt{2}$ as shown in Fig. \ref{fig1}.
\begin{figure}
  \centering
  \includegraphics[width=\figurewidth]{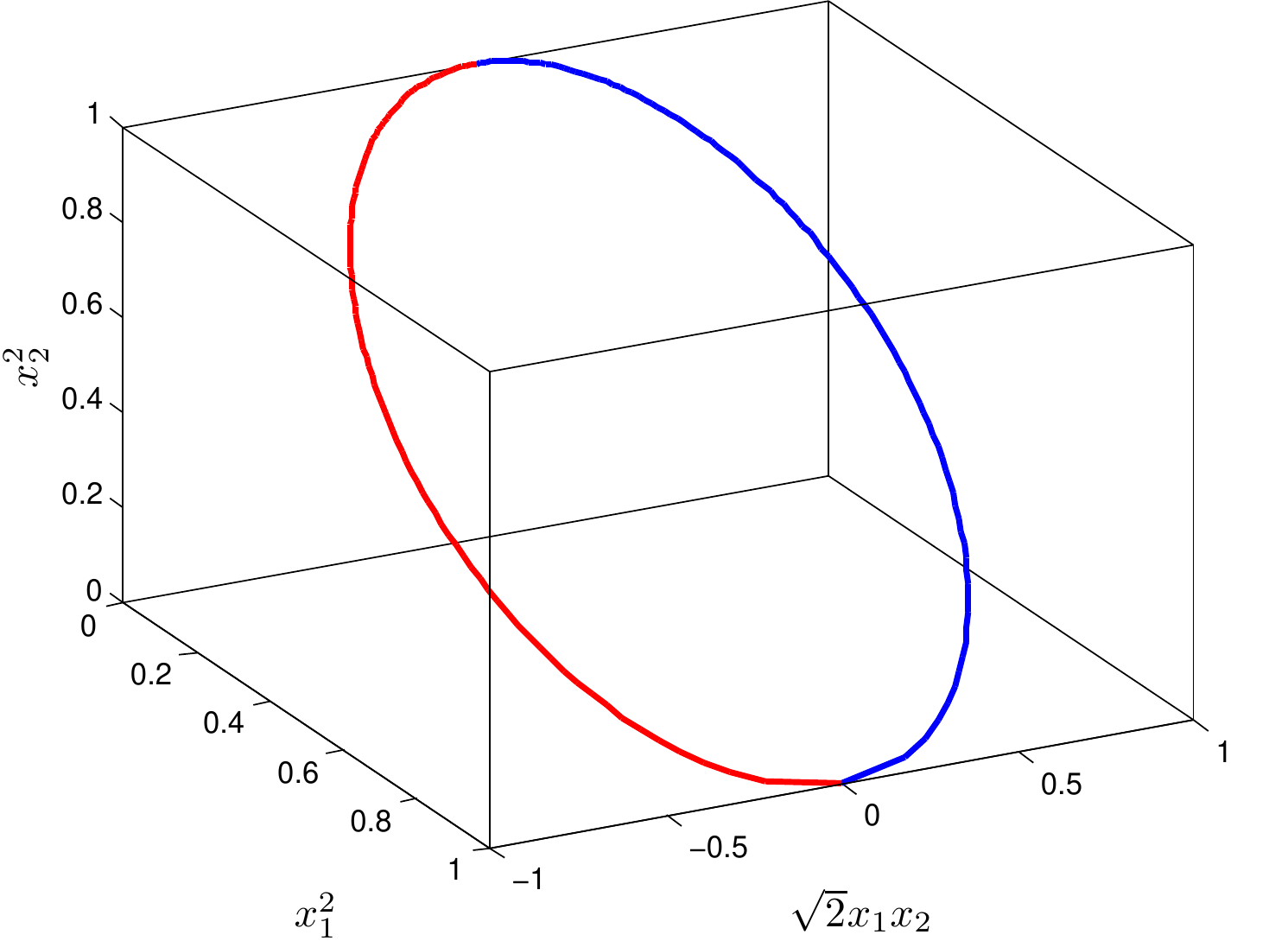}
  \caption{$\mathcal{P}_{2,1}$ can be illustrated as a circle with radius $1/\sqrt{2}$.}
  \label{fig1}
\end{figure}

\subsection{Tangent space and normal space of $\mathcal{P}_{N,s}$}
In this part, we study the first order differential structure of $\mathcal{P}_{N,s}$. The reason we need it is that the condition number, which is an important quantity used for the stable embedding of a manifold, will be defined by the normal bundle.

Denote $\mathrm{skew}(N)$ as the set of all $N\times N$ skew symmetric matrices, and $\mathrm{sym}(N)$ as the set of all $N\times N$ symmetric matrices. In the following lemma, the tangent space and the normal space at every point of $\mathcal{P}_{N,s}$ are unveiled. Remind that by stating $P_X=XX^\T$, it is indicated that $X\in\GL_{N,s}$.

\begin{lemma}
The tangent space of $\mathcal{P}_{N,s}$ at a point $P_X=XX^\T$ is
\begin{align}\label{tangent}
\mathcal{T}_{P_X}\mathcal{P}_{N,s}=\{X_\bot KX^\T + XK^\T X_\bot^\T:\;K\in\mathbb{R}^{(N-s)\times s}\},
\end{align}
and the normal space is
\begin{align}\label{normal}
\mathcal{N}_{P_X}\mathcal{P}_{N,s}=\mathrm{skew}(N)\bigcup
\{S\in \mathrm{sym}(N):\;S=(P_X-P_{X_\bot})S_0,S_0\in \mathrm{sym}(N)\}.
\end{align}
$X_\bot$ is the matrix such that $[X,X_\bot]^\T [X,X_\bot]= I_N$.
\end{lemma}

\begin{proof}
Denote the sets in \eqref{tangent} and \eqref{normal} as $U$ and $V$, respectively, both of which are subspaces.

Suppose that $P_{X(t)}=X(t)X(t)^\T$ is a curve on $\mathcal{P}_{N,s}$ with $P_{X(0)}=P_{X}$. The tangent vector along such curve at point $P_{X}$ is
\begin{align*}
\frac{\mathrm{d}}{\mathrm{d}t}\bigg|_{t=0} P_{X(t)}
= \frac{\mathrm{d}X(t)}{\mathrm{d}t}\bigg|_{t=0} X(0)^\T + X(0) \frac{\mathrm{d}X(t)^\T}{\mathrm{d}t}\bigg|_{t=0} = \xi X^\T + X \xi^T,
\end{align*}
in which $\xi$ is a vector in $\mathcal{T}_{X}\GL_{N,s}$. Because these tangent vectors all belong to the tangent space of $\mathcal{P}_{N,s}$ at $P_X$, we have
\begin{align*}
\mathcal{T}_{P_X}\mathcal{P}_{N,s}\supset U=\{\xi X^\T+X\xi^\T:\;\xi\in\mathcal{T}_{X}\GL_{N,s}\}.
\end{align*}
To check that $U\bot V$, notice that any symmetric matrix is orthogonal to any skew-symmetric matrix, and that
\begin{align*}
&\mathrm{tr}((X_\bot KX^\T + XK^\T X_\bot^\T)(XX^\T-X_\bot X_\bot^\T)S_0)\\
= \,&\mathrm{tr}((X_\bot K X^\T - XK^\T X_\bot^\T)S_0)\\
= \,&\mathrm{tr}(X_\bot K X^\T S_0^\T) - \mathrm{tr}(S_0XK^\T X_\bot^\T)=0.
\end{align*}
The dimension of $V$ is
\begin{align*}
N(N-1)/2+s(s-1)+(N-s)(N-s-1)/2+N
=N^2-Ns+s^2,
\end{align*}
and the dimension of $U$ is $Ns-s^2$, which is equal to the dimension of $\mathcal{T}_{P_X}\mathcal{P}_{N,s}$. Therefore, we finish the proof that $\mathcal{T}_{P_X}\mathcal{P}_{N,s}=U$ and $\mathcal{N}_{P_X}\mathcal{P}_{N,s}=V$.
\end{proof}

\subsection{Condition Number of $\mathcal{P}_{N,s}$}
For the purpose of delineating the regularity of a manifold, a notion called the condition number of a manifold, also known as the reach of a manifold \cite{RanProjMan}, is introduced.

\begin{definition}\cite{RanProjManIMPORTANT}
Let $\mathcal{M}$ be a compact Riemannian sub-manifold of $\mathbb{R}^N$. The condition number is defined as $1/\tau$, where $\tau$ is the largest number having property that the open normal bundle about $\mathcal{M}$ of radius $r$ is embedded in $\mathbb{R}^N$ for all $r<\tau$.
\end{definition}

The condition number $1/\tau$ controls both local properties, such as the curvature of any unit-speed geodesic curve on the manifold, and global properties, such as how close the manifold may curve back upon itself at long geodesic distance \cite{HomologyRanSam}. However, from its definition, the condition number of a manifold is not easy to obtain in general. In this section, we focus on the condition number of the manifold $\mathcal{P}_{N,s}$. In the next section, we shall see that the following lemma is a key step to the main result.

\begin{lemma}\label{conprojm}
The condition number $1/\tau$ of $\mathcal{P}_{N,s}$ is $\sqrt{2}$.
\end{lemma}
\begin{proof}
The proof is based on the fact that $\tau$ is the radius of the largest possible non-self-intersecting tube around the manifold.

The distance between $P_X$ and the set of skew matrices is $\|P_X-0\|_F=\sqrt{s}$, in that
\begin{align*}
\|P_X\|_F&=\sqrt{\mathrm{tr}(XX^\T XX^\T)}
=\sqrt{\mathrm{tr}(X^\T X)}=\sqrt{s}.
\end{align*}
Thus, if the tube around the manifold intersects itself at a point of skew symmetric matrix, then the radius of the tube is no less than $\sqrt{s}$.

For any $X$ and $Y$ in $\GL_{N,s}$, denote $\bar{X}=XX^\T-X_\bot X_\bot^\T$ and $\bar{Y}=YY^\T-Y_\bot Y_\bot^\T$. If the tube around the manifold intersects itself at a point of symmetric matrix $\Phi$, then $\exists X, Y\in \GL_{N,s}$ such that
\begin{align}\label{intersection}
\bar{X}S_1=\bar{Y}S_2=\Phi \in \mathrm{sym}(N).
\end{align}
Equation \eqref{intersection} is equivalent to the condition that $\Phi$ and $\bar{X}$ share the same eigenvector matrix $\widetilde{X}$, and $\Phi$ and $\bar{Y}$ share the same eigenvector matrix $\widetilde{Y}$, which is not necessarily to be the same as $\widetilde{X}$. One of the eigenspaces of $\bar{X}$ is $\mathrm{span}(X)$ corresponding to eigenvalue $1$, and the other one is $\mathrm{span}(X_\bot)$ corresponding to eigenvalue $-1$. Thus, $\widetilde{X}$ is composed of basis of $\mathrm{span}(X)$ and basis of $\mathrm{span}(X_\bot)$. Without loss of generality, we can assume that $\widetilde{X}=[X,X_\bot]$ and $\widetilde{Y}=[Y,Y_\bot]$. Now we are able to prove that if $\|\Phi-P_X\|_F<1/\sqrt{2}$ and $\|\Phi-P_Y\|_F<1/\sqrt{2}$, then \eqref{intersection} can not hold.

First we consider the case where $\mathrm{rank}(\Phi)=s$. Suppose that \eqref{intersection} holds, and $\|\Phi-P_X\|_F<1/\sqrt{2}$ and $\|\Phi-P_Y\|_F<1/\sqrt{2}$.
According to the discussion in the previous paragraph, we have that
$\|\Phi-XX^\T\|_F = \| \Lambda_1 - \Lambda\|_F$ and
$\|\Phi-YY^\T\|_F = \| \Lambda_2 - \Lambda\|_F,$
in which $\Lambda$, $\Lambda_1$, and $\Lambda_2$ are diagonal matrices. $\Lambda = \mathrm{diag}[1,\cdots,1,0,\cdots,0]$, in which the number of $1$ is $s$. Both $\Lambda_1$ and $\Lambda_2$ have $s$ non-zero elements and $N-s$ zeros on the diagonal. Observe that if the non-zeros of $\Lambda_1$ is the first $s$ elements on its diagonal, then the non-zeros of $\Lambda_2$ can not be the first $s$ elements on its diagonal. Otherwise, $X$ and $Y$ span the same subspace. Thus, a contradiction comes from the fact that at least one of $\| \Lambda_1 - \Lambda\|_F>1$ and $\| \Lambda_2 - \Lambda\|_F>1$ holds.

The second case is  $\mathrm{rank}(\Phi)<s$. If \eqref{intersection} holds, then it is obvious that $\| \Lambda_1 - \Lambda\|_F>1$ and $\| \Lambda_2 - \Lambda\|_F>1$, so $\|\Phi-P_X\|_F<1/\sqrt{2}$ and $\|\Phi-P_Y\|_F<1/\sqrt{2}$ can not hold.

The third case is $\mathrm{rank}(\Phi)>s$. Suppose that the eigenvalues of $\Phi$ are $\lambda_1\geq\lambda_2\geq\cdots\geq\lambda_N$. If \eqref{intersection} holds and $\|\Phi-P_X\|_F<1/\sqrt{2}$, then we must have $\lambda_1\geq\cdots\geq\lambda_s>1/2>\lambda_{s+1}\geq\cdots\geq\lambda_N,$
and eigenvalues $\lambda_1,\cdots,\lambda_s$ correspond to the eigenspace $\mathrm{span}(X)$. Since span$(X)\neq$span$(Y)$, $$\| \Lambda_2 - \Lambda\|_F>\sqrt{(1/2)^2+(1-1/2)^2}=\sqrt{1/2},$$
so $\|\Phi-P_Y\|_F<1/\sqrt{2}$ can not hold.

These three cases show that $\|\Phi-P_X\|_F<1/\sqrt{2}$, $\|\Phi-P_Y\|_F<1/\sqrt{2}$, and \eqref{intersection} can not hold simultaneously. From the discussion in the third case, it is obvious that $\|\Phi-P_X\|_F=1/\sqrt{2}$, $\|\Phi-P_Y\|_F=1/\sqrt{2}$, and \eqref{intersection} can hold simultaneously by choosing
$
\lambda_1=\cdots=\lambda_{s-1}=1, \lambda_s=\lambda_{s+1}=1/2,
\lambda_{s+2}=\cdots=\lambda_{N}=0.
$

Consequently, for any $P_X,P_Y\in\mathcal{P}_{N,s}$, $\Phi\in\mathbb{R}^{N\times N}$, if $\|\Phi-P_X\|_F<1/\sqrt{2}$ and $\|\Phi-P_Y\|_F< 1/\sqrt{2}$, their normal spaces can not intersect at point $\Phi$. If $\|\Phi-P_X\|_F=1/\sqrt{2}$ and $\|\Phi-P_Y\|_F= 1/\sqrt{2}$, then there exists a $\Phi$ at which their normal spaces intersect. Thus, for $\mathcal{P}_{N,s}$, $\tau=1/\sqrt{2}$, and the condition number $1/\tau$ is $\sqrt{2}$.
\end{proof}

The condition number $1/\tau$ of $\mathcal{P}_{N,s}$ provides the regularity of this manifold, so its RIP is able to be derived.

\section{Proof of Main Result}
\label{sec:proof}
\begin{proof}
Basically, Theorem \ref{RIPP} is proved by applying the condition number of the manifold of the projection matrix $\mathcal{P}_{N,s}$, calculating the covering number of the set of chords of $\mathcal{P}_{N,s}$, and utilizing the Johnson-Lindenstrauss lemma.

The set of chords of a manifold $\mathcal{M}$ is denoted as
\begin{align}
\mathcal{C}(\mathcal{M}):&
=\left\{\frac{X-Y}{\|X-Y\|_F}:\;X,Y\in\mathcal{M},\;X\neq Y\right\}.
\end{align}
From lemma C.1 in \cite{StaManEmb}, we know that for any $0<T\leq\ 3\tau /4$, the set $\mathcal{C}(B_T)$ is a $(4\sqrt{T/\tau},\|\cdot\|_2)$-cover of $\mathcal{C}(\mathcal{P}_{N,s})$, where $B_T$ is defined as
\begin{align}
B_T=\bigcup_{P\in\mathcal{N}(\mathcal{P}_{N,s},T)} \{P+\mathcal{T}_P\mathcal{P}_{N,s}(T)\},
\end{align}
in which $\mathcal{N}(\mathcal{P}_{N,s},T)$ is the $(T,d_g)$-cover of $\mathcal{P}_{N,s}$, and $\mathcal{T}_P\mathcal{P}_{N,s}(T):=\{\xi\in\mathcal{T}_P\mathcal{P}_{N,s}:\;\|\xi\|_F\leq T\}$. It is easily shown in Part B of the proof of Theorem III.1 in \cite{StaManEmb} that the $\epsilon$-cover of $\mathcal{C}(B_T)$ satisfies
\begin{align}
|\mathcal{N}(\mathcal{C}(B_T), \epsilon)| \leq \nonumber |\mathcal{N}(\mathcal{P}_{N,s},T)|\left(1+\frac{2}{\epsilon}\right)^s
+|\mathcal{N}(\mathcal{P}_{N,s},T)|^2\left(1+\frac{2}{\epsilon}\right)^{2s+1}.
\end{align}

Theorem 8 in \cite{MetricEntropyHomoSp} gives that the $(T,d_p)$ covering number of Grassmann manifold $\mathrm{Gr}_{N,s}$ is $\left(C_0/T\right)^{s(N-s)}$ with $C_0$ being a universal constant.
Remind that the projection distance on $\mathrm{Gr}_{N,s}$ is defined as
\begin{align*}
d^2_p(X,Y):=\frac{1}{2} \|XX^\T - YY^\T\|_F^2 .
\end{align*}
Thus, the $(\sqrt{2}T,d)$ covering number of $\mathcal{P}_{N,s}$ is also $\left(C_0/T\right)^{s(N-s)}$.

For any $Z\in\mathcal{C}(\mathcal{P}_{N,s})$, there exist $\widetilde{Z}\in\mathcal{C}(B_T)$ and $\hat{Z}\in\mathcal{N}(\mathcal{C}(B_T), \epsilon)$ such that $\|(Z-\widetilde{Z})\|_2\leq 4\sqrt{T/\tau}$ and $\|(\hat{Z}-\widetilde{Z})\|_2\leq \varepsilon$. We then have
\begin{align}\label{lessthan}
\|\mathcal{A}(Z)\|_2 &\leq \|\mathcal{A}(Z-\widetilde{Z})\|_2 + \|\mathcal{A}(\widetilde{Z}-\hat{Z})\|_2 +\|\mathcal{A}(\hat{Z})\|_2\nonumber\\
&\leq 4\sqrt{T/\tau}+\epsilon+\|\mathcal{A}(\hat{Z})\|_2
\end{align}
and
\begin{align}\label{greaterthan}
\|\mathcal{A}(Z)\|_2 &\geq \|\mathcal{A}(\hat{Z})\|_2 - \|\mathcal{A}(Z-\hat{Z})\|_2\nonumber\\
&\geq \|\mathcal{A}(\hat{Z})\|_2 - (\|\mathcal{A}(Z-\widetilde{Z})\|_2 + \|\mathcal{A}(\widetilde{Z}-\hat{Z})\|_2)\nonumber\\
&\geq \|\mathcal{A}(\hat{Z})\|_2 - 4\sqrt{T/\tau}-\epsilon.
\end{align}
Equations \eqref{lessthan} and \eqref{greaterthan} together with Lemma 1.1 in \cite{RanProjManIMPORTANT} (known as the Johnson-Lindenstrauss lemma) give that
\begin{align}\label{sup}
\sup_{Z\in\mathcal{C}(\mathcal{P}_{N,s})} \|\mathcal{A}(Z)\|_2 \nonumber
&\leq \sup_{\hat{Z}\in\mathcal{N}(\mathcal{C}(B_T), \epsilon)} \epsilon+4\sqrt{T/\tau} + \|\mathcal{A}(\hat{Z})\|_2\\
&\leq \epsilon+4\sqrt{\frac{T}{\tau}}+(1+\delta)\sqrt{\frac{m}{N^2}},
\end{align}
and
\begin{align}\label{inf}
\inf_{Z\in\mathcal{C}(\mathcal{P}_{N,s})} \|\mathcal{A}(Z)\|_2
&\geq \inf_{ \hat{Z}\in\mathcal{N}(\mathcal{C}(B_T), \epsilon)} \|\mathcal{A}(\hat{Z})\|_2-(\epsilon+4\sqrt{T/\tau}) \nonumber\\
&\geq  (1-\delta)\sqrt{\frac{m}{N^2}}- \epsilon-4\sqrt{\frac{T}{\tau}}.
\end{align}
Equations \eqref{sup} and \eqref{inf} hold simultaneously with probability exceeding $1-|\mathcal{N}(\mathcal{C}(B_T), \epsilon)|^{-\beta}$ given that
\[m\geq \left( \frac{4+2\beta}{\delta^2/2-\delta^3/3}\right)\log |\mathcal{N}(\mathcal{C}(B_T), \epsilon)|.\]
From Lemma \ref{conprojm}, we know that $\tau=1/\sqrt{2}$. Let $\delta=\varepsilon/2$, $\epsilon=\varepsilon\sqrt{m}/(4N)$, and $T=\tau m\varepsilon^2/(256N^2)$, then
\begin{align*}
|\mathcal{N}(\mathcal{C}(B_T), \epsilon)|
&\leq\left(\frac{\sqrt{2}C_0}{T}\right)^{s(N-s)}\left(1+\frac{2}{\epsilon}\right)^s 
+ \left(\frac{\sqrt{2}C_0}{T}\right)^{2s(N-s)}\left(1+\frac{2}{\epsilon}\right)^{2s+1}\\
&\leq 2 \left(\frac{512N^2C_0}{m\varepsilon^2}\right)^{2s(N-s)}
\left(1+8\sqrt{\frac{N^2}{m\varepsilon^2}}\right)^{2s+1}\\
&\leq 2 \left(\frac{512N^2C_0}{\varepsilon^2}\right)^{2s(N-s)}
\left(1+8 \frac{N}{\varepsilon}\right)^{2s+1}.
\end{align*}
Thus, we conclude that if
\begin{align*}
m\geq&\left( \frac{4+2\beta}{\varepsilon^2/8-\varepsilon^3/24}\right)
\bigg(\log 2+ 2s(N-s)\log\left(\frac{512N^2C_0}{\varepsilon^2}\right)
+(2s+1)\log\left(1+\frac{8N}{\varepsilon}\right)\bigg)\\
\sim &\left( \frac{2+\beta}{\varepsilon^2-\varepsilon^3/3}\right) O\left(s(N-s)\log \left(\frac{N}{\varepsilon}\right)\right)
\end{align*}
then $\forall Z\in\mathcal{C}(\mathcal{P}_{N,s})$,
$(1-\varepsilon)\sqrt{m}/N \leq \|\mathcal{A}(Z)\|_2 \leq (1+\varepsilon)\sqrt{m}/N$
holds with probability exceeding $1-|\mathcal{N}(\mathcal{C}(B_T), \epsilon)|^{-\beta}$.

In order to control the probability above, we need a lower bound on $|\mathcal{N}(\mathcal{C}(B_T), \epsilon)|$. Notice that
\begin{align*}
|\mathcal{N}(\mathcal{C}(B_T), \epsilon)|&\geq |\mathcal{N}(\mathcal{P}_{N,s},T)|\left(1+\frac{2}{\epsilon}\right)^s\\
&\geq|\mathcal{N}(\mathcal{P}_{N,s},T)|
=\left|\mathcal{N}\left(\GL_{N,s},\frac{T}{\sqrt{2}}\right)\right|\geq\left(\frac{\sqrt{2}c_0}{T}\right)^{s(N-s)},
\end{align*}
in which the last inequality holds when $T\leq \sqrt{2}\pi/4$ according to Theorem 8 in\cite{MetricEntropyHomoSp}, and $c_0$ is a universal constant. Because $T=\tau m\varepsilon^2/ (256N^2)< \tau/256= 1/(256\sqrt{2})$, we have
$$
|\mathcal{N}(\mathcal{C}(B_T), \epsilon)|\geq (512c_0)^{s(N-s)}=\mathrm{exp}(c_1 s(N-s)).
$$
Thus, the probability exceeding
$$
1- |\mathcal{N}(\mathcal{C}(B_T), \epsilon)|^{-\beta} \geq 1-\mathrm{exp}(-c_1\beta s(N-s)).
$$
\end{proof}

\section{Conclusion}

Subspace projection matrices are low rank matrices with additional structure that allows for further reduction of its intrinsic dimension. In this work, the restricted isometry property of subspace projection matrix under random orthonormal compression is studied.

The set of $s$ dimensional subspace projection matrices $\mathcal{P}_{N,s}$ is modeled as an $s(N-s)$ dimensional submanifold in $\mathbb{R}^{N\times N}$, so the main concern is transformed to the problem of the stable embedding of $\mathcal{P}_{N,s}$ into $\mathbb{R}^{N\times N}$. One of the key points is the calculation of the conditional number $1/\tau$ of $\mathcal{P}_{N,s}$. Once $\tau$ is obtained, the RIP is able to be established by applying covering sets of the set of chords of $\mathcal{P}_{N,s}$ and utilizing the JL lemma. In order to calculate the condition number $1/\tau$, the tangent space and the normal space at every point of $\mathcal{P}_{N,s}$ are investigated.
The result is that by $O(s(N-s)\log N)$ measurements the RIP of subspace projection matrix is guaranteed.

This work is not exhausted. The condition number $1/\tau$ of $\mathcal{P}_{N,s}$ provides the regularity of this manifold, so its RIP under other random compressions, such as i.i.d. Gaussian compression, could also be established using such condition number. Furthermore, theoretical analysis on algorithms for compressed subspace projection matrix recovery could be built upon this work.

\bibliographystyle{unsrt}
\bibliography{refs1}

\begin{thebibliography}{10}

\bibitem{candes2008intro}
E.~J. Cand{\`e}s and M.~B. Wakin.
\newblock An introduction to compressive sampling.
\newblock {\em Signal Processing Magazine, IEEE}, 25(2):21--30, 2008.

\bibitem{Fazel2002}
M.~Fazel.
\newblock {\em Matrix Rank Minimization with Applications}.
\newblock PhD thesis, Stanford University, CA, USA, 2002.

\bibitem{baraniuk2008simple}
R.~Baraniuk, M.~Davenport, R.~DeVore, and M.B. Wakin.
\newblock A simple proof of the restricted isometry property for random
  matrices.
\newblock {\em Constructive Approximation}, 28(3):253--263, 2008.

\bibitem{LRRIPNNM}
B.~Recht, M.~Fazel, and P.~Parrilo.
\newblock Guaranteed minimum-rank solutions of linear matrix equations via
  nuclear norm minimization.
\newblock {\em SIAM Review}, 52(3):471--501, 2010.

\bibitem{LRRIP}
E.~J. Cand\`{e}s and Y.~Plan.
\newblock Tight oracle inequalities for low-rank matrix recovery from a minimal
  number of noisy random measurements.
\newblock {\em Information Theory, IEEE Transactions on}, 57(4):2342--2359,
  April 2011.

\bibitem{RIPCS}
E.~J. Cand\`{e}s.
\newblock The restricted isometry property and its implications for compressed
  sensing.
\newblock {\em Comptes Rendus Mathematique}, 346(9):589 -- 592, 2008.

\bibitem{LSFC}
K.~C. Lee, J.~Ho, and D.~Kriegman.
\newblock Acquiring linear subspaces for face recognition under variable
  lighting, May 2005.

\bibitem{MoSeg}
S.~Rao, R.~Tron, R.~Vidal, and Y.~Ma.
\newblock Motion segmentation in the presence of outlying, incomplete, or
  corrupted trajectories, Oct 2010.

\bibitem{1315111}
J.~Ho, K.~C. Lee, M.~H. Yang, and D.~Kriegman.
\newblock Visual tracking using learned linear subspaces.
\newblock In {\em Computer Vision and Pattern Recognition, 2004 IEEE Computer
  Society Conference on}, volume~1, pages I--782--I--789 Vol.1, June 2004.

\bibitem{OITSHII}
L.~Balzano, R.~Nowak, and B.~Recht.
\newblock Online identification and tracking of subspaces from highly
  incomplete information.
\newblock In {\em Communication, Control, and Computing (Allerton), 2010 48th
  Annual Allerton Conference on}, pages 704--711, Sept 2010.

\bibitem{PETRELS}
Y.~Chi, Y.~C. Eldar, and R.~Calderbank.
\newblock Petrels: Subspace estimation and tracking from partial observations.
\newblock In {\em Acoustics, Speech and Signal Processing (ICASSP), 2012 IEEE
  International Conference on}, March 2012.

\bibitem{CSCCaseStudy}
X.~Mao and Y.~Gu.
\newblock Compressed subspace clustering: A case study.
\newblock In {\em Global Conference on Signal and Information Processing, 2014
  IEEE}, Dec 2014.

\bibitem{StaManEmb}
H.~L. Yap, M.~B. Wakin, and C.~J. Rozell.
\newblock Stable manifold embeddings with structured random matrices.
\newblock {\em Selected Topics in Signal Processing, IEEE Journal of},
  7(4):720--730, Aug 2013.

\bibitem{ManEmbNew}
A.~Eftekhari and M.~B. Wakin.
\newblock New analysis of manifold embeddings and signal recovery from
  compressive measurements.
\newblock {\em CoRR}, abs/1306.4748, 2013.

\bibitem{ManRanLProj}
M.~A. Iwen and M.~Maggioni.
\newblock Approximation of points on low-dimensional manifolds via random
  linear projections.
\newblock {\em CoRR}, abs/1204.3337, 2012.

\bibitem{RanProjMan}
K.~L. Clarkson.
\newblock Tighter bounds for random projections of manifolds.
\newblock In {\em Symposium on Computational Geometry}, pages 39--48. ACM,
  2008.

\bibitem{RanProjManIMPORTANT}
R.~G. Baraniuk and M.~B. Wakin.
\newblock Random projections of smooth manifolds.
\newblock {\em Foundations of Computational Mathematics}, 9(1):51--77, 2009.

\bibitem{HomologyRanSam}
P.~Niyogi, S.~Smale, and S.~Weinberger.
\newblock Finding the homology of submanifolds with high confidence from random
  samples.
\newblock {\em Discrete and Computational Geometry}, 39(1-3):419--441, 2008.

\bibitem{OptAl}
P.~A. Absil, R.~Mahony, and R.~Sepulchre.
\newblock {\em Optimization Algorithms on Matrix Manifolds}.
\newblock Princeton University Press, 2008.

\bibitem{Coin}
D.~Achlioptas.
\newblock Database-friendly random projections: Johnson-lindenstrauss with
  binary coins.
\newblock {\em J. Comput. Syst. Sci.}, 66(4):671--687, 2003.

\bibitem{RiePurMatRec}
M.~Tan, I.~W. Tsang, L.~Wang, B.~Vandereycken, and S.~J. Pan.
\newblock Riemannian pursuit for big matrix recovery.
\newblock In {\em Proceedings of the 31st International Conference on Machine
  Learning (ICML)}, volume~32 of {\em JMLR Workshop and Conference
  Proceedings}, pages 1539--1547, 2014.

\bibitem{JohnLee}
J.~M. Lee.
\newblock {\em Introduction to Smooth Manifolds}.
\newblock Springer, 2003.

\bibitem{GrassGeom}
P.~A. Absil, R.~Mahony, and R.~Sepulchre.
\newblock Riemannian geometry of grassmann manifolds with a view on algorithmic
  computation.
\newblock {\em Acta Applicandae Mathematica}, 80(2):199--220, 2004.

\bibitem{MetricEntropyHomoSp}
S.~J. Szarek.
\newblock Metric entropy of homogeneous spaces.
\newblock In {\em in Quantum Probability, Banach Center Publ. 43, Polish Acad.
  Sci}, pages 395--410, 1998.

\end{thebibliography}
\end{document}